\documentclass[%
 reprint,
 amsmath,amssymb,
 aps,
]{revtex4-2}

\usepackage{graphicx}
\usepackage{dcolumn}
\usepackage{bm}

\usepackage{amsmath}
\usepackage{amssymb}
\usepackage{dsfont}
\usepackage{amsthm}
\usepackage{color}
\usepackage{array}
\usepackage{bm}
\usepackage{hyperref}
\usepackage{lipsum}
\usepackage{mathtools}

\theoremstyle{theorem}
\newtheorem{theorem}{Theorem}
\theoremstyle{remark}
\newtheorem{lemma}{Lemma}

\newcommand{\expec}[1]{\mathbb{E}\big[#1\big]}

\newcommand{\p}{\mathds{P}}

\newcommand{\real}{\mathbb{R}}

\newcommand{\bin}{\operatorname{Binomial}}
\newcommand{\ber}{\operatorname{Bernoulli}}

\begin{document}

\preprint{APS/123-QED}

\title{Limited Parallelization in Gate Operations Leads to Higher Space Overhead and Lower Noise Threshold}

\author{Sai Sanjay Narayanan,  Smita Bagewadi and Avhishek Chatterjee}
\affiliation{
Indian Institute of Technology, Madras, Chennai, India
}


\date{\today}

\begin{abstract}
In a modern error corrected quantum memory or circuit, parallelization of  gate operations is severely restricted due to issues like cross-talk \cite{Zhao2022quantum, Arute2019quantum, Sarovar2020detecting}. Hence, there are enough idle qubits not undergoing gate operations either during the computation phase or during the error correction phase, which suffer further decoherence while waiting.  Thus, in reality, the space overhead and the noise threshold would depend on the level of gate parallelization. In this paper, we obtain an analytical lower bound on the required space overhead in terms of the level of parallelization for an error correction framework that has more error correction capability than the existing ones. We consider two types of errors: i.i.d. erasure and depolarization. In comparison to the known lower bounds which assume full gate parallelization \cite{harrow2003robustness,razborov2004upper,kempe2008upper,fawzi2022lower,UthirakalyaniNC2023}, our bound is provably strictly larger despite allowing more capability to the error correction framework. This shows the steep price to be paid for lack of gate parallelization. An implication of the bound is that the noise or decoherence threshold, i.e., the noise beyond which no fault-tolerant memory or circuit can be realized, vanishes if the number of parallel gate operations does not scale linearly with the number of physical qubits. 

 \end{abstract}

 \maketitle

Fault-tolerant quantum computation happens in two periodically alternating phases: computation phase and error correction phase \cite{aharonov1997fault,gottesman2014fault,fawzi2020constant}. In most fault-tolerance models, qubit decoherence and gate errors during the computation phase are modeled together as a noise operator, but the decoherence of qubits during the error correction phase is generally ignored \cite{aharonov1997fault,gottesman2014fault,fawzi2020constant,fawzi2022lower,UthirakalyaniNC2023}. 
However, in practice, qubit decoherence in the error correction phase can be significant, due to the lack of parallelization in gate operations. 

In the error correction phase, (syndrome) measurements and corrective gate operations on qubits, like Pauli gates for stabilizer codes, take non-negligible time \cite{Plantenberg2007demonstration,Dicarlo2009demonstration,Isenhower2010demonstration}. In current technologies, these  operations  cannot be fully parallelized due to issues like cross-talks \cite{Zhao2022quantum, Arute2019quantum, Sarovar2020detecting} . Hence, these operations can be performed only in batches, resulting in enough idle qubits in each batch which interact with the bath and decohere further. Naturally, the size of the batch would depend on the technology and material used for implementing the circuit \cite{Zhao2022quantum, Arute2019quantum, Sarovar2020detecting}. On the other hand, the rate of decoherence of waiting qubits also depends on the material and  technology.

In this paper, we characterize the interplay and impact of the rate of qubit decoherence and the level of parallelization in gate operations on fault-tolerance. In particular, we characterize the required space overhead in terms of the rate (strength) of qubit decoherence and the amount of parallelization in gate operations during error correction. It is known that a lower bound on the space overhead of quantum circuits of a certain depth is implied by a lower bound on the space overhead of maintaining a quantum memory for a time equal to that circuit depth \cite{fawzi2022lower}. Hence, though our bounds directly  apply to quantum circuits, in this paper we mostly talk about fault-tolerant quantum memories.


Our main contribution is the derivation of a new lower bound on the space overhead requirement for a quantum memory when there is qubit decoherence and the parallelization in gate operations is limited.  The bound on required overhead depends explicitly on the amount of parallelization in gate operations and the rate of decoherence, and is strictly more than the best known lower bounds when there is no restriction on gate parallelization \cite{harrow2003robustness,razborov2004upper,kempe2008upper, fawzi2022lower,UthirakalyaniNC2023}. For example, one of the best known lower bound on space overhead is $\frac{1}{Q(\mathcal{N})}$, where $Q(\mathcal{N})$ is the quantum capacity of the quantum channel $\mathcal{N}$ which represents noise or errors \cite{nielsen2002quantum,khatri2020principles}. The observed increase in the space overhead is the penalty paid for the lack of parallelization in gate operations. Further, for any non-zero decoherence rate, we prove that quantum states stored using any code is unrecoverable after a constant time (independent of number of physical qubits), if the amount of parallelization in gate operations is below a threshold. This threshold depends on the decoherence rate. To the best of our knowledge, this is the first work to analytically characterize the effect of (lack of) gate parallelization on fault-tolerance overhead, which is an important issue for any practical quantum circuit \cite{Zhao2022quantum, Arute2019quantum, Sarovar2020detecting}.


\subsection*{{\bf Gate Parallelization and Space Overhead}}
Consider a quantum memory that stores a quantum state of $l$ logical qubits using $n$ physical qubits. The $l$ qubit logical state is encoded using some quantum code and the encoding is assumed to be perfect. This memory is periodically error corrected and at the end of $T$ error correction cycles the original $l$ qubit state is retrieved using a perfect decoder. The periodic error correction scheme, a.k.a. the fault-tolerance scheme, is said to keep the memory error corrected for $T$ cycles if the retrieved state is within a prescribed distance from the original $l$ qubit state. The space overhead of this scheme is defined as $\frac{n}{l}$. In this work we are interested in the asymptotic range $l \to \infty$.

In an error corrected memory there are two phases which occur periodically: static or rest phase and error correction phase. A static phase is of length $\bar{t}_{st}$ and an error correction phase is of length $\bar{t}_{ec}$.

During the static phase $\bar{t}_{st}$, qubits undergo decoherence due to interactions with the environment. In this work, we model decoherence as i.i.d. $q$ erasure noise or depolarizing noise, i.e., a qubit state is either erased or it becomes $\frac{I}{2}$ with probability $q$. In general, $q$ depends on the material, technology and $\bar{t}_{st}$. 

The above model is also true for a quantum circuit \cite{fawzi2022lower,UthirakalyaniNC2023}. In the context of a quantum circuit, $\bar{t}_{st}$ is the computation phase, and the noise in $\bar{t}_{st}$ captures both qubit decoherence and gate error. An error correction phase follows a computation phase. Keeping a memory corrected for $T$ cycles is equivalent to ensuring close to accurate computation of depth $T$.

 An error correction phase in a stabilizer code based error-corrected memory has two parts: syndrome extraction followed by Pauli correction \cite{aharonov1997fault,nielsen2002quantum}. However, we do not restrict to this widely used  stabilizer error correction framework. Rather, we consider a more general error correction framework by assuming more capability than the stabilizer framework. Since we obtain information theoretic converse bounds on  space overhead and noise threshold in this framework, these bounds are applicable to all known fault-tolerance schemes including the ones based on stabilizer codes. 
 
 We assume that the locations of the erasures and depolarizations ($\frac{I}{2}$) are known to the fault-tolerance scheme in each error correction phase, but the perfect decoder, which intends to retrieve the original state at the end of $T$ cycles, only knows the error statistics. 
 Thus, in this framework, the required gate operations in the error correction phase are minimal: correct errors by applying gates at the known locations.  Furthermore, we assume the following higher capability of the error correction phase: one gate operation on an erroneous qubit is enough for correcting it. Obviously, no existing or future fault-tolerance scheme can correct an erroneous using less than one gate operation.

Error correction phase has multiple batches, due to limitations on gate parallelization. A batch is of duration $t_g$ and at most $k= n\alpha$ qubits can undergo gate operations in a batch. Qubits that are not part of the batch remain idle and decohere independently (erasure or depolarization)  with probability $p$ that depends on the material, technology and $t_g$. The locations of the decohered idle qubits are known instantaneously. We allow even more capability to the fault-tolerance scheme by assuming that idle qubits decohering during a batch can be potentially corrected in that batch while respecting the parallelization constraint and qubits which are part of a batch do not decohere during the duration of that batch. The decision rule to include a subset of (erroneous) qubits in a batch does not prioritize any qubit.

Clearly, as mentioned before, the above error correction framework has more error correction capability than any currently known fault-tolerance scheme. 
Hence, the information theoretic converse bounds that we derive for  space overhead and noise threshold in this framework are directly applicable to all existing and arguably, future fault-tolerance schemes. There is scope for tightening our bounds for the  currently feasible schemes by assuming more modest error correction capabilities of the above framework. However, as presented later, our converse bounds derived for this more capable error correction framework are able to bring out the penalty in space overhead and noise threshold for the lack of gate parallelization. Tightening of the converse bounds for an error correction framework with more modest error correction capabilities is left for future work.

The following theorem gives a lower bound on space overhead in terms of the level of parallelization and the noise level in the correction phase. It also provides a necessary condition on the required level of parallelization in gate operations.

\begin{theorem}
\label{thm:overheadThreshold}
For $0<\theta<\frac{p-\alpha}{p}$, there exists a constant $C(\alpha,p,\theta)$, which is independent of $l$ and non-increasing in $\theta$, such that  for keeping the above memory  error corrected with a non-vanishing probability  against i.i.d. erasures    
for a time duration $C(\alpha,p,\theta)$ 
any fault-tolerance scheme must have $n \ge \frac{l p}{2 \alpha - p + 2p \theta}$ when $\alpha \in (\frac{p}{2},p)$ and $q\to 0$. Moreover, no fault-tolerance scheme can keep the memory corrected for arbitrary constant time (independent of $l$) with a non-vanishing probability  
if the level of parallelization,  $\alpha<\frac{p}{2}$. 

A similar lower bound is true for depolarizing noise for the same $C(\alpha,p,\theta)$: $n \ge \frac{l}{Q_{\mbox{dep}}\left(\frac{p-\alpha}{p}-\theta\right)}$ for $\frac{2p}{3}<\alpha<p$. Fault-tolerance for arbitrary constant time (independent of $l$) is impossible if $\alpha<\frac{2p}{3}$.


The above lower bounds on space overhead and level of parallelization are non-decreasing in $q$.

\end{theorem}



 Later in this section, we discuss an outline of the proof of the above theorem followed by a mean-field intuition behind the result. First, we discuss the intuitive meaning of this theorem. 
 
 The theorem states that, for erasure (depolarizing) noise, given $p$ and $\alpha>\frac{p}{2}$ ($\frac{2p}{3}$), to ensure that the state in the memory remains accurate for a constant time $C(\alpha,p,\theta)$ the lower bound on  space overhead decreases with $\theta$. Thus, $\theta$ is a tunable parameter, which captures required increase in space overhead when the memory needs to be maintained for a (constant) longer time. By choosing $\theta$ to be sufficiently small, we observe that when $n < \frac{l p}{2\alpha-p}$ $\left(n < \frac{l}{Q_{\mbox{dep}}\left(\frac{p-\alpha}{p}\right)} \right)$, fault-tolerant memory for a duration more than a constant $C(p,\alpha,\theta)$ is impossible for i.i.d. $p$ erasure (depolarizing) in the error correction phase.  

 The theorem also states that no fault-tolerance technique is useful in maintaining the state in the memory against i.i.d. erasure (depolarizing) noise for an arbitrary constant time when $\alpha<\frac{p}{2}$ ($\frac{2p}{3}$). This implies that $\frac{p}{2}$ ($\frac{2p}{3}$) to be lower bound on the level of parallelization for ensuring fault-tolerance over an arbitrary constant time.

The above bounds are stated in the theorem for $q \to 0$. The purpose is to bring out the effect of errors during the error correction phase, even when errors in the static phase of a memory (computation phase of a circuit) are negligible. The above result shows that errors in the error correction phase can hurt fault-tolerance severely when there is lack of parallelization in gate operations. However, the lower bounds in Theorem \ref{thm:overheadThreshold} on space overhead and parallelization level are true for any $q>0$. This is because the space overhead bounds are monotonic in $q$, as stated at the end of the theorem.

 As mentioned above, the lower bounds in Theorem \ref{thm:overheadThreshold} are applicable to fault-tolerant quantum circuits as well. In that case, time is equivalent to circuit depth. In quantum computation, almost all noiseless quantum circuits on $l$ logical qubits have depths that scale with $l$. In fact, in practice most scalings are faster than $\log l$. Since the lower bounds in the theorem  apply directly to constant depth circuits (equivalent to constant time memory), these lower bounds are also true for all useful circuits when $l$ is large enough.

 Some practical implications of Theorem \ref{thm:overheadThreshold} is discussed in the next section. Now we discuss an outline of the proof of the above theorem and its mean-field intuition.

An important step in the proof is to show that within a constant time (depending only on $p$ and $\alpha$, and not $l$), $n\frac{p-\alpha}{p}$ number of uncorrected errors accumulate with probability close to $1$ when $\alpha<p$ and $q \approx 0$. For showing the above bound we use the fact that when errors are erasure or depolarlization, the state of an uncorrected qubit does not change if it decoheres further. It is also observed  that the accumulated uncorrected errors are uniformly distributed.
This is because the errors are i.i.d. and the error correction batches do not follow any priority among erroneous qubits for error correction. 


The lower bound on the space overhead for $\alpha<p$ is obtained by obtaining an upper bound on the capacity of a quantum channel with uniform $n\frac{p-\alpha}{p}$ number of errors. For this step, we show that the capacity of channels with $n\gamma$ uniform errors is upper bounded the capacity of a channel with i.i.d. errors with probability $\gamma$. Finally, we use the capacity of i.i.d. erasure and depolarizing channels for obtaining the lower bound on space overhead.

The lower bounds on the parallelization $\alpha$ is obtained by using the same approach as above.  Here, we use the fact that capacity of i.i.d $\gamma$ erasure (depolarizing) channels are zero for $\gamma>\frac{1}{2}$ $\left(\gamma>\frac{1}{3}\right)$. Hence, the memory under i.i.d $p$ erasure (depolarizing) noise is information theoretically impossible to recover after a constant time if $\frac{p-\alpha}{p}>\frac{1}{2}$ $\left(\frac{p-\alpha}{p}>\frac{1}{3}\right)$.

As mentioned before, a crucial technical part of the proof of Theorem~\ref{thm:overheadThreshold} is the derivation of the high probability bound on the time by which $n\frac{p-\alpha}{p}$ uncorrected errors accumulate. This is obtained by developing a high probability bound on the hitting time of a well chosen non-birth-death Markov chain. The Markov chain is chosen to lower bound  the evolution of errors in a stochastic dominance sense. Obtaining a closed-form lower bound on even the expected hitting time of a non-birth-death Markov chain is a non-trivial challenge \cite{MarkovChainHittingTime}. We obtain the high probability bound on the hitting time by recursively moving a stochastic dominance bound and exploiting the specific structure of the concerned Markov chain.

We refer to Appendix for a rigorous derivation of the high probability bound on the time by which $n\frac{p-\alpha}{p}$ uncorrected errors accumulate. This derivation is presented inside the complete proof of  Theorem~\ref{thm:overheadThreshold}. Here, we discuss an intuitive mean-field derivation of the bound on the time by which $n\frac{p-\alpha}{p}$ uncorrected errors accumulate.


Suppose the mean number of uncorrected errors after $t$th error correction phase is ${x}_t$. Then, the mean number of uncorrected errors after $(t+1)$th phase would increase by 
$\max((n-x_t) p - n \alpha,0)$. Here, the first term corresponds to the average number of qubits that were fine after $t$, but got erased after $t+1$, and $n \alpha$ is the maximum number of erroneous qubits that could be corrected at $(t+1)$th phase due to the constraint on parallel gate operations. Note that if the number of erroneous qubits at phase $0$ is $0$, the mean number of uncorrected errors is an increasing sequence and hence, it must converge \cite{Rosenlicht1986analysis}. At the steady state, the increase in the sequence should saturate, implying $(n-x_{\infty}) p - n \alpha = 0$. Hence, the number of uncorrected qubits at steady state would be $x_{\infty}=n \frac{p-\alpha}{p}$.

Note that the mean number of uncorrected errors is is growing by at least $n\epsilon p$ in each phase till it reaches $n\left(\frac{p-\alpha}{p}-\epsilon\right)$, it . Thus, the number of phases to reach $n\left(\frac{p-\alpha}{p}-\epsilon\right)$ is upper bounded by $\frac{p-\alpha}{\epsilon p^2}$, a constant that depends on $p$ and $\alpha$, but not on $n$ or $l$. 

Interestingly, the above mean field derivation matches the rigorous bounds in a qualitative sense. For a continuous time Markov chain on $N$ states can be provably approximated by a mean field differential equation when $N \to \infty$ under certain conditions on the Markov chain. A class of theorems known as the Kurtz's theorems \cite{Kurtz1978strong} formalize such results. However, we are not aware of such results for provable approximations of discrete time Markov chains using discrete time difference equations. Here, we observe a qualitative closeness of the above mean-field arguments with the rigorous bound, whose proof is in Appendix. Intuitively, due to measure concentration \cite{BoucheronLM2013} of i.i.d. errors, when the number of physical qubits is large, the change in the number of uncorrected errors between two consecutive error correction phases is close to the mean value with a probability that is exponentially close to $1$. Hence, by the application of the union bound, over a constant  (independent of number of physical qubits) number of phases, the evolution of total uncorrected errors is close to the evolution of the above mean field evolution.

\subsection*{{\bf Implications}}

 There are some interesting implications of Theorem~\ref{thm:overheadThreshold} that are of practical relevance. As discussed after Theorem~\ref{thm:overheadThreshold}, when $q\to 0$, the lower bounds on the space overheads, $\frac{lp}{2\alpha-p}$ and $\frac{l}{Q_{\mbox{dep}}\left(\frac{p-\alpha}{p}-\theta\right)}$, increase as $\alpha$ decreases over the range $(\frac{p}{2},p)$ and $(\frac{2p}{3},p)$, for erasure  and deploarizing noise, respectively. Clearly, for $q\to 0$, the corresponding lower bounds without any constraint on the level of parallelization \cite{fawzi2022lower,UthirakalyaniNC2023} are $\frac{l}{1-2p}$ and $\frac{l}{Q_{\mbox{dep}}(p)}$. The lower bounds in Theorem~\ref{thm:overheadThreshold} are strictly larger than the bounds without parallelization constraint when $\alpha<p(1-p)$. Thus, a direct implication of Theorem~\ref{thm:overheadThreshold} is that  lack of parallelization provably hurts in terms of the space overhead if parallelization is below a certain level ($p(1-p)$).

Theorem ~\ref{thm:overheadThreshold} also implies that fault-tolerance using quantum codes is impossible if the level of parallelization, $\alpha$, is less than 
$\frac{p}{2}$ and $\frac{2p}{3}$, for erasure and depolarizing noise respectively. 
These parallelization thresholds, in turn, imply upper bounds on noise threshold that are strictly lower than the known bounds. The best known upper bounds on thresholds for $p$ (when $q\to 0$) for the case  with fully parallel gate operations \cite{fawzi2022lower, UthirakalyaniNC2023}  are $\frac{1}{2}$ and $\frac{1}{3}$ for erasure and depolarizing noise, respectively. On the other hand, when level of parallelization is $\alpha$, the respective noise thresholds are $2\alpha$  and $\frac{3\alpha}{2}$. This means that the upper bound on the respective noise thresholds approach zero when the level of gate parallelization  approaches zero. This, in turn, implies that fault-tolerance is impossible if the number of parallel operations in a batch, $k$,  does not scale linearly with $n$.

In the case with fully parallelized gate operations, for $q>0$, the overhead bound is  $n \ge \frac{l}{1-2(1-(1-p)(1-q))}$ when errors are erasures \cite{fawzi2022lower,UthirakalyaniNC2023}. This is because, in the fully parallelized case, one needs only one correction batch to correct all errors including the ones that occur during that correction batch. Hence, the effective errors are distributed as i.i.d $1-(1-p)(1-q)$. After some algebraic calculations it can be seen that the lower bound on space overhead in the above theorem (for $q\to 0$)  is strictly greater than $\frac{1}{1-2(1-(1-p)(1-q))}$ if $\alpha<p(1-p)(1-q)$. It turns out that the same condition is true for depolarizing noise. In practice $p$ and $q$ are mostly less than $0.1$. This implies that when $\alpha<0.8~p$, the space overhead needed for $q \to 0$ is significantly more than the overhead needed for $q>0$ in the presence of full gate parallelization. This clearly shows that lack of gate parallelization in the error correction phase  cost more in terms of overhead than higher error rate ($q>0$) in the static (computation) phase.

Theorem~ \ref{thm:overheadThreshold} can also offer some guidelines for  practical implementations of quantum circuits. The probability of decoherence $p$ during an error correction batch of length $t_g$ depends on the duration as well as the nature of the material and the environment which dictate the rate of decoherence with time. A well accepted model for Markov decoherence is $p=1-\exp(-\kappa t_g)$ \cite{nielsen2002quantum,khatri2020principles}, where $\kappa$ is the rate of decoherence. Under this decoherence model, the theorem directly gives a lower bound on the minimum level of gate parallelization, $\alpha$, that is necessary for a given decoherence rate $\kappa$: $\frac{1-\exp(-\kappa{t}_g)}{2}$ and $\frac{1-\exp(-\kappa{t}_g)}{1.5}$ for erasure and depolarizing noise, respectively. Thus, if we plot $\frac{1-\exp(-\kappa{t}_g)}{2}$ (or $\frac{1-\exp(-\kappa{t}_g)}{1.5}$) as a surface in $3$D ($\kappa$, $t_g$ and $\alpha$), any technology with ($\kappa, t_g, \alpha$) lying below that surface can never realize a fault-tolerant circuit or memory. In other words, any quantum hardware technology must strive to be above the surface.

The required space overhead for erasure noise is $\frac{1-\exp(-\kappa t_g)}{2\alpha-1+\exp(-\kappa t_g)}$ when $\alpha>\frac{1-\exp(-\kappa{t}_g)}{2}$. When $\kappa t_g$ is sufficiently small, which is expected in a good hardware, the overhead becomes $\left(\frac{2\alpha}{\kappa t_g} - 1\right)$. Hence, the ratio,  $\frac{\alpha}{\kappa t_g}$, must be kept higher than $\frac{1}{2}$ for having any kind of fault-tolerance. A technology with higher value of this ratio can employ error correction with lower overhead.

\subsection*{{\bf Conclusion}}
Modern quantum technologies like superconducting qubits suffer from cross-talk which severely restricts parallelization in gate operations. In this paper,  we started an analytical investigation into the basic aspects of the limitations imposed by the lack of parallelization in gate operations. Our  results  show that the limitation in gate parallelization  leads to higher space overhead. Limited gate parallelization can also result in vanishing noise threshold if the number of parallel gate operations does not scale linearly with the number of physical qubits. In fact, decreased gate parallelization can lead to worse overhead than increased noise in the static (computation) phase of a memory (circuit). Our quantitative bounds offer some guidelines on the conditions to be followed in the designing of practical systems. We would like to emphasize that this work is simply a first step towards quantifying the penalty to be paid due to lack of gate parallelization. A thorough investigation into this issue would be crucial for practical realization of fault-tolerant quantum systems.

In future,  it would be important to understand the interplay between gate parallelization and correlated noise. Since in practical quantum systems, correlated noise is an important issue. Also, this work implicitly models cross-talk as a sharp  binary function of parallelization: there is no cross-talk till a certain level of parallelization, which we model as $\alpha$ and very high cross-talk beyond that. It would be interesting to investigate the effect of a more smooth dependence of the intensity of cross-talk on the level of parallelization.

\bibliographystyle{IEEEtran}
\bibliography{bibfile}


\appendix

\section{Proof of Theorem \ref{thm:overheadThreshold}}
\label{sec:proof}
Since $q\to 0$, in this more capable error correction framework, the sources of errors are the error correction batches. Since $\bar{t}_{st}$ and $t_{ec}$ are constants independent of $l$, if the memory degrades after a constant number of error correction phases, it can be said that the memory degraded after a finite time. Since each error correction phase contains at least one batch, it is enough to show that the memory degrades after a constant number of batches. Also, since there is almost no error in $\bar{t}_{st}$, we can ignore the static phases and see the system as a sequence of error correction batches, which we call error correction epochs or epochs.

After an error correction epoch, each correct qubit undergoes decoherence with probability $p$. On the other hand, the qubits that were erased or depolarized after the previous error correction stay that way. During the error correction phase,  the correct locations of the errors are known. Since our goal is a lower bound on space overhead, as discussed before, allowing  more capability is fine. Note that in the correction phase, at most $k =  n \alpha$ qubits can undergo gate operations, and hence, at most $k=n \alpha$ qubits can be corrected. This leads to the following Markov dynamics for the number of errors.

\subsection*{Markov Chain of Errors}

The evolution of the number of errors is modeled as a Markov process as follows. Let $X_t$ denote the number of corrupted (decohered) qubits at the start of time epoch $t$. Initially, there are no corrupted qubits ($X_0 = 0$). At the start of time epoch $t$, due to interactions with the environment, each qubit that is not already decohered is acted on by a noise operator (e.g., erasure or depolarizing noise), making it decohere with probability $p$. If the noise acts on an already decoherence qubit, it stays the same, due to the nature of erasure and depolarizing noise operators.

Since noise acts independently on each qubit and , the number of qubits that decohere at time epoch $t$ is given by a $\bin(n - X_t, p)$ random variable, which we denote as $Y_t$. Then, in the correction phase, we can correct $f_t$ of the qubits, where $f_t = f(X_t + Y_t) = \min\{ X_t + Y_t, n\alpha\} $. 
The resulting quantum state at the start of time epoch $t+1$ has $X_{t+1} = X_t + Y_t - f_t$ decohered qubits, and the sequence $\{X_t\}_{t \in \mathbb{N}}$ forms a Markov chain.

\subsection*{Preliminaries}
In this section, we will state some results that will be used throughout the paper. First, we look at the familiar concentration bounds for binomial random variables.
\begin{lemma}{(Concentration bounds)}
\label{lemma1_conc}
    Let $X \sim \bin(n, p)$. Then,
    \begin{align*}
        \p[X \geq n(p + \epsilon)] &\leq \exp{(-2n\epsilon^2)}, \text{ and} \\
        \p[X \leq n(p - \epsilon)] &\leq \exp{(-2n\epsilon^2)}.
    \end{align*}
\end{lemma}
\begin{proof}
    By definition, $X$ can be written as a sum of $n$ independent $\ber(p)$ random variables, each of which is $1/2-$subgaussian (by Hoeffding's lemma). Therefore, $X$ is $\sqrt{n}/2-$subgaussian, and we get the concentration bounds by applying Chernoff's inequality and using the subgaussian property \cite{BoucheronLM2013}.
\end{proof}
Next, we mention an important property of the Markov chain $\{X_t\}$, related to the monotonicity of $\p[X_{t+1} \geq k \mid X_t = x]$ with respect to $x$; if $x$ is more, then intuitively this probability should be higher.
\begin{lemma}
\label{lemma2_mon}
    Let $h_k(x) = \p[X_{t+1} \geq k \mid X_t = x]$. Then, for all $k \in [0,n]$, $h_k(\cdot)$ is an increasing function.
\end{lemma}
\begin{proof}
    We have $h_k(x) = \p[x + Y_t - f(x + Y_t) \geq k \mid X_t = x] $.
    Consider the function $g(z) = z - f(z)$. From the definition of $f$, we observe that $g$ is non-decreasing, and hence for each $k \in [0,n]$ there exists a unique $z_k$ (given by $\inf_z \{g(z) \geq k \}$) such that $g(z) \geq k \Longleftrightarrow z \geq z_k $. Thus, we can write $h_k(x) = \p[x + Y_t \geq z_k \mid X_t = x] = \p[Y_t \geq z_k - x \mid X_t = x] = \p[Z \geq z_k - x]$, where $Z \sim \bin(n-x, p)$. \\
    We can rewrite this probability as $h_k(x) = \p[n-x-Z \leq n-x - (z_k-x)] = \p[Z' \leq n - z_k]$, where $Z' = n-x-Z \sim \bin(n-x, 1-p)$. We may think of $Z'$ as a sum of $n-x$ independent $\ber(1-p)$ random variables. It is now evident that if $x$ increases, then $n-x$ decreases, and thus $h_k(x) = \p[Z' \leq n - z_k]$ increases (not necessarily strictly, as $z_k - x$ could be negative), because $n-z_k$ is a constant independent of $x$. 
\end{proof}

As a corollary, we have the following extension of the previous lemma.
\begin{lemma}
\label{lemma3_ext_mon}
    Let $h_k^{(m)}(x) = \p[X_{t+m} \geq k \mid X_t = x]$, where $m \geq 1$. Then, for all $k \in [0,n]$, $h_k^{(m)}(\cdot)$ is an increasing function. Further, this implies that $\p[X_t \geq k]$ is increasing in $t$.
\end{lemma}
\begin{proof}
    We prove by induction on $m$. We have already proved the result for $m = 1$ in lemma \ref{lemma2_mon}. Now, suppose the result is true for $m = j$. Consider
    \begin{align*}
        &\p[X_{t+j+1} \geq k \mid X_t = x] \\      
        &\begin{multlined}
            =\sum_{y = 0}^n \p[X_{t+j+1} \geq k \mid X_{t+j} = y, X_t = x] \\ \cdot \p[X_{t+j} = y \mid X_t = x ]
        \end{multlined}\\
        &= \sum_{y=0}^n \p[X_{t+j+1} \geq k \mid X_{t+j} = y] \cdot \p[X_{t+j} = y \mid X_t = x ].
    \end{align*}
    Denote $a_y = \p[X_{t+j+1} \geq k \mid X_{t+j} = y]$ and $b_y = \p[X_{t+j} = y \mid X_t = x]$. Then, we have
    \begin{align*}
        &\p[X_{t+j+1} \geq k \mid X_t = x] = a_0 b_0 + a_1 b_1 + \hdots + a_n b_n \\
        &\begin{multlined}
            = a_0 \left(\sum_{i=0}^n b_i\right) + (a_1 - a_0)\left(\sum_{i=1}^n b_i\right) \hdots +(a_n - a_{n-1})b_n
        \end{multlined}\\
        &\begin{multlined}
            = a_0 \cdot \p[X_{t+j} \geq 0 \mid X_t = x]\\ \qquad \qquad+ \sum_{i=1}^n (a_i - a_{i-1}) \cdot p[X_{t+j} \geq i \mid X_t = x].
        \end{multlined}
    \end{align*}
    By the induction hypothesis, each $\p[X_{t+j} \geq i \mid X_t = x] $ is increasing in $x$, and we also have $a_0 \leq a_1 \hdots \leq a_n$ by lemma \ref{lemma2_mon}. Thus, the above sum is a non-negative linear combination of increasing functions, and hence $\p[X_{t+j+1} \geq k \mid X_t = x]$ is increasing in $x$. By induction, the result is true for all $m \geq 1$. \\
    Now, consider for any $t \geq 0$,
    \begin{align*}
        &\p[X_{t+1} \geq k] \\
        &= \sum_{x = 0}^n \p[X_{t+1} \geq k \mid X_1 = x] \cdot \p[X_1 = x] \\
        &\geq \sum_{x = 0}^n \p[X_{t+1} \geq k \mid X_1 = 0] \cdot \p[X_1 = x] \\
        &=  \p[X_{t+1} \geq k \mid X_1 = 0] \\
        &= \p[X_t \geq k \mid X_0 = 0] \quad = \p[X_t \geq k],
    \end{align*}
    and hence $\p[X_{t} \geq k]$ is increasing in $t$.
\end{proof}

\subsection*{Hitting Time Bound}
Suppose we have a quantum error correcting code that can correct up to $n\beta$ errors. Then, it is a crucial requirement that $X_t \leq n\beta$ after a sufficiently long time, as otherwise we cannot reliably store the system in memory. The quantities that we are most interested in analyzing are $\p[X_t > n\beta]$, and $\expec{\tau}$, where $\tau = \inf_t\{X_t > n\beta\}$ denotes the first time that the number of decohered qubits is more than $n\beta$. In this paper, we only focus on the former.\\
We will prove that if $\beta$ is smaller than a certain threshold, then $\{X_t > n\beta\}$ happens with high probability in constant time. Before proving this result, we state a few helpful lemmas.

\begin{lemma}
\label{lemma4_helper}
Suppose $\beta <  \frac{p - \alpha}{p}$, and let $x \leq n\beta$. Then, we have for any $\delta > 0$,
\begin{multline}
    \p[X_{t+1} > x + (n-x)(p - \delta) - n\alpha \mid X_t = x] \\ \geq 1 - \exp(-2n(1-\beta)\delta^2).
\end{multline}
\end{lemma}
\begin{proof}
    Recall that, conditioned on $X_t = x$, $Y_t \sim \bin(n-x, p)$ is the number of qubits that decohere at time $t$. Clearly, if $Y_t > (n-x)(p-\delta)$, we have $X_{t+1} > x + (n-x)(p - \delta) - n\alpha$ because at most $n\alpha$ qubits are corrected in the correction phase. Therefore,
    \begin{align}
        &\p[X_{t+1} > x + (n-x)(p - \delta) - n\alpha \mid X_t = x] \\
        &\geq \p[Y_t > (n-x)(p - \delta) \mid X_t =x ] \\
        &= \p[\bin(n-x, p) > (n-x)(p - \delta)] \\
        &\geq 1 - \exp(-2(n-x)\delta^2) \hspace{4mm} (\text{from lemma \ref{lemma1_conc}} ) \\
        &\geq 1 - \exp(-2n(1-\beta)\delta^2) \hspace{4mm} (\because x \leq n\beta).
    \end{align}
\end{proof}

\begin{lemma}
    \label{lemma5}
    Let $g:\real \longmapsto \real$ be defined as $ g(x) = x + (n-x)(p - \delta) - n\alpha$, and let $\{x_k\}_{k \geq 0}$ be the sequence defined by $x_0 = 0, $ and $x_{k+1} = g(x_k)$ for all $k \geq 0$. Then,
    \begin{equation}
        x_{k} = \frac{ n(p - \delta - \alpha) }{p - \delta} \big(1 - (1-p+\delta )^k \big), \text{  for all }k \geq 1.
    \end{equation}
    In particular, for $0 < \delta < p - \alpha$, the sequence $\{x_k\}_{k\geq0}$ is monotonically increasing, and converges to the value $\frac{ n(p - \delta - \alpha) }{p - \delta}$.
\end{lemma}
\begin{proof}
    Note that $g(x) = Ax + B$, where $A = (1-p+\delta), B = n(p-\delta -\alpha)$. We show by induction that $x_k = B(A^{k-1} + \hdots + A + 1) $ for all $k \geq 1$. Clearly, the statement is true for $k=1$ because $x_1 = A\cdot 0 + B = B$. Now suppose it is true for some $k \geq 1$. Consider
    \begin{align}
        x_{k+1} &= A x_k + B \\
        &= B(A^k + \hdots + A^2 + A) + B \\
        &= B(A^k + \hdots + A + 1),
    \end{align}
    and thus the statement is true for $k+1$ as well. By induction, the statement is true for all $k \geq 1$. In particular, we may write $A^{k-1} + \hdots + A + 1 = \frac{1 - A^k}{1 - A}$, to get $x_k = \frac{B}{1-A}(1 - A^k)$, which gives us the desired formula.

    For $0 < \delta < p - \alpha$, we have $B = n(p-\alpha-\delta) > 0$ and $0< A=1-p+\delta < 1$, and thus $(1-A^k)$ is monotonically increasing and converges to $1$ $\implies \{x_k\}$ is monotonically increasing and converges to $\frac{B}{1-A} = \frac{n(p-\delta - \alpha)}{p - \delta}$.
\end{proof}
Now we derive the main result. Let $\beta < \frac{p-\alpha}{p}$, and choose any $\delta$ that satisfies $0 < \delta < p - \frac{\alpha}{1 - \beta}$. Define the sequence $\{x_k\}_{k \geq 0} $ as in lemma \ref{lemma5}. Then, we have $\delta - \delta \beta < p - p \beta - \alpha \implies \beta < \frac{p-\alpha-\delta}{p-\delta} = \frac{1}{n} \lim_{k \rightarrow \infty} x_k$. Let $T \in \mathbb{N}$ be the smallest number such that $x_k > n\beta$ for all $k \geq T$; such a $T$ exists because of lemma \ref{lemma5}, and more importantly it does not depend on $n$. We can compute $T$ as follows:
\begin{align}
     &\frac{ n(p - \delta - \alpha) }{p - \delta} \big(1 - (1-p+\delta )^k \big) > n\beta \\
     \Longleftrightarrow& (1-p+\delta)^k < 1 - \frac{\beta(p - \delta)}{p - \delta - \alpha} \\
     \Longleftrightarrow& k > \frac{ \log(\big( p - \alpha - p\beta - \delta(1-\beta)\big ) - \log(p - \alpha - \delta)}{\log(1-p+\delta)},
\end{align}
and thus $T = \Big\lceil \frac{ \log\big( p - \alpha - p\beta - \delta(1-\beta)\big ) - \log(p - \alpha - \delta)}{\log(1-p+\delta)} \Big\rceil $. Note that $T$ depends on $p, \alpha, \beta$, as well as our choice of $\delta$. We may fix $\delta = \frac{1}{2} \left( p - \frac{\alpha}{1 - \beta} \right)$, in which case we obtain
\begin{equation}
    T(p,\alpha, \beta) = \left\lceil \frac{ \log\left( \frac{p - \alpha - p\beta}{2}\right) - \log\left(\frac{p}{2} - \alpha + \frac{\alpha}{2(1-\beta)}\right)}{\log\left(1-\frac{p}{2} - \frac{\alpha}{2(1-\beta)}\right)} \right\rceil. 
\end{equation}
We now state the following theorem.
\begin{theorem}
\label{thm:hittingTimeWHP}
Suppose $\beta < \frac{p-\alpha}{p}$, and define $T = T(p,\alpha,\beta)$ as above. Then, we have
\begin{equation}
    \p[X_t > n\beta] \geq \big( 1 - \exp(-2n(1-\beta)\delta^2)\big)^T,
\end{equation}
for all $t \geq T$. \\
In other words, for a given $\beta < \frac{p-\alpha}{p} $, the number of erroneous qubits crosses $n\beta$ in constant time (independent of $n$) with high probability (in the limit $n \longrightarrow \infty$).
\end{theorem}
\begin{proof}
Let $\{x_k\}_{k \geq 0}$ be the sequence as defined in lemma \ref{lemma5}. For any $1 \leq t \leq T$, we have $x_{t-1} \leq n\beta$ by definition. Hence, consider
\begin{align}
    &\p[X_t > x_t] \\
    &\geq \p[X_t > x_t , X_{t-1} > x_{t-1}] \\
    &= \sum_{x > x_{t-1}} \p[X_t > x_t \mid X_{t-1} = x] \cdot \p[X_{t-1} = x] \\
    &\begin{multlined}
        \geq \sum_{x > x_{t-1}} \p[X_t > x_t \mid X_{t-1} = x_{t-1}] \cdot \p[X_{t-1} = x] \\ \hspace{3mm} (\text{from lemma \ref{lemma2_mon}})
    \end{multlined}\\
    &\begin{multlined}
        \geq  \sum_{x > x_{t-1}} \big(1 - \exp(-2n(1-\beta)\delta^2)  \big) \cdot \p[X_{t-1} \geq x] \\ \hspace{3mm} (\text{from lemma \ref{lemma4_helper}}) 
    \end{multlined}\\
    &=  \big(1 - \exp(-2n(1-\beta)\delta^2)  \big) \cdot \p[X_{t-1} > x_{t-1}],
\end{align}
and since $\p[X_1 > x_1] = \p[X_1 > x_1 \mid X_0 = 0 = x_0] > 1 - \exp(-2n(1-\beta)\delta^2)$, we have from the above recursion,
\begin{equation}
    \p[X_t > x_t] \geq \big(1 - \exp(-2n(1-\beta)\delta^2)  \big)^t,
\end{equation}
for all $1 \leq t \leq T$. In particular, substituting $t = T$ and using the fact that $x_T > n\beta$, we have
\begin{equation}
    \p[X_T > n\beta] \geq \p[X_T > x_T] \geq \big(1 - \exp(-2n(1-\beta)\delta^2)  \big)^T,
\end{equation}
and hence from lemma \ref{lemma3_ext_mon}, it follows that $\p[X_t > n\beta] \geq \big(1 - \exp(-2n(1-\beta)\delta^2)  \big)^T$ for all $t \geq T$, as desired.

\end{proof}

\subsection*{Obtaining the bounds in Theorem \ref{thm:overheadThreshold}}
This proof builds on the bound in Theorem~\ref{thm:hittingTimeWHP}. Since $\beta<\frac{p-\alpha}{p}$ let us take $\beta=\frac{p-\alpha}{p}-\theta$. After  $T$ epochs (where $T \equiv T(p, \alpha, \beta)$ is independent of $l$), Theorem \ref{thm:hittingTimeWHP} tells us that the probability of having more than $n\beta$ errors is lower bounded by $(1-\exp(-c_1 n))^T$. This quantity is arbitrarily close to $1$ for sufficiently large $l$ since $n\ge l$.


Since the errors are i.i.d. and the gate operations do not prioritize among errors, these $n\beta$ errors are uniformly distributed across $n$ qubits. Therefore, in order to recover memory after time $T(p, \alpha, \beta)$, it is necessary to have a quantum code that can correct up to $n\beta$ uniform errors. For $\theta<\frac{p-\alpha}{p}$, for all $t \ge C(\alpha,p,\theta)=T(p, \alpha, \frac{p-\alpha}{p}-\theta)$ (independent of $l$) and for all sufficiently large $l$ (and hence $n$, since $n \ge l$), there will be at least $n\beta = n\left(\frac{p-\alpha}{p}-\theta\right)$ uniform errors with probability almost $1$. Hence,  in order to recover memory after an arbitrarily long (but constant) time, we must have a quantum code that can correct up to $n\left(\frac{p-\alpha}{p}-\theta\right)$ uniform errors.


Note that both for erasure and depolarizing noise, if an error occurs on a qubit the state becomes erasure or $\frac{I}{2}$. A further error on this qubit does not change the state.   In the following proof, this fact is used implicitly. This allows us to track the number of errors and not worry about the nature of the errors. Since the nature of errors are all the same, either erasure or $\frac{I}{2}$. This is not true for general error models. We first consider erasure errors. The proof for depolarizing is almost similar.

Consider an i.i.d. quantum erasure channel with probability of erasure $\beta-\gamma$. Suppose that there is a code with rate strictly more than $1-2\beta$, that can correct uniform $n\beta$ erasures. For any $\gamma>0$, and for all sufficiently large $n$, the number of erasures in i.i.d. erasure channel is upper-bounded by $n\beta$ with probability $1-\exp(-2 n\gamma^2)$, by lemma \ref{lemma1_conc}. Further, by the property of i.i.d. Bernoulli coin toss, given that a certain number of erasures have happened, the erasures are uniformly distributed. Thus, the code for uniform errors can recover a state with probability $1-\exp(-2 n \gamma^2)$ over the i.i.d. $\beta-\gamma$ erasure channel. This implies that we have a code with rate strictly more than $1-2\beta$ for an i.i.d. $\beta-\gamma$ erasure channel. By taking $\gamma$ arbitrarily close to zero, we see that the capacity result of i.i.d. erasure channels is violated. Thus, by contradiction, we get that $n\beta$ uniform erasures cannot be corrected using a code with rate more than $1-2\beta$. Therefore, any code that can correct $n\beta$ uniform erasures must have overhead of at least $\frac{1}{1 - 2\beta}=\frac{p}{2\alpha-p+2\theta p}$.

{Since the quantum capacity of an i.i.d. $\beta$ erasure channel vanishes at $\beta=\frac{1}{2}$, we cannot have a quantum code  that can correct more than $\frac{n}{2}$ uniform erasures. Therefore, if $\frac{p-\alpha}{p} -\theta > \frac{1}{2}$ for any $\theta>0$ we cannot keep the memory corrected for a time more than $C(\alpha,p,\theta)$, which is a constant independent of $l$. In particular, if $\alpha<\frac{p}{2}$, i.e., $\frac{p-\alpha}{p}>\frac{1}{2}$, there exists a $\theta>0$ such that $\frac{p-\alpha}{p}-\theta>\frac{1}{2}$. In this case, no quantum error correction scheme can keep the memory corrected for any time greater than $C(\alpha,p,\theta)$. This proves that when $\alpha<\frac{p}{2}$ no fault-tolerance scheme can keep the memory corrected for an arbitrarily large time.}

A similar argument gives the result for depolarising noise to lead to the lower bound for keeping memory corrected for a time $C(\alpha,p,\theta)$: $n \ge \frac{l}{Q_{\mbox{dep}}\left(\frac{p-\alpha}{p}-\theta\right)}$ for $\frac{2p}{3}<\alpha<p$. The parallelization lower bound (on $\alpha$) is obtained by observing the fact that quantum capacity of i.i.d. $\beta$ depolarizing channel vanishes at $\beta=\frac{1}{3}$.

\subsection*{Bounds are monotonic in $q$}
In the more capable error correction framework, consider two different memories with decoherence probabilities during static phase being $q$ and $q'$ respectively, but with the same $p$ for the error correction phase, and let $q'>q>0$. We model errors as i.i.d. during static phase as well as error correction phase. We consider the following error generation model for $q$ which is stochastically coupled with the i.i.d. error generation model of $q'$-memory. 

The errors generated in each error correction phase of $q$-memory are the same as that of $q'$-memory. On the other hand, for every error in every static phase of $q'$-memory, a coin with probability of head $\frac{q}{q'}$ is tossed independently. If it is head, the same error is created in $q$-memory at the same location.

It can be seen that under this coupled error generation model, the error distributions of  $q$-memory and $q'$-memory remain faithful to their original  distributions. However, on each sample path (realization) their errors become coupled \cite{MarkovChainHittingTime}. 

Clearly, $q'$-memory have more errors than $q$ on each sample path. Thus, if a code allow accurate retrieval from $q'$-memory, then the same code will allow accurate retrieval from $q$-memory. Therefore, if no  code with overhead $<x$ can correct $q$-memory, no code with overhead $<x$ can correct $q'$-memory. This shows that the lower bounds on space overhead are monotonic in $q$. A similar argument implies that the lower bounds on the required level of parallelization are monotonic in $q$.

\end{document}